\documentclass[a4paper]{amsart}
\usepackage[french,german,english]{babel}
\usepackage{graphicx}
\usepackage{amscd}
\usepackage{amsmath}
\usepackage{amsfonts}
\usepackage{amssymb}
\usepackage{bbm}
\usepackage{setspace}
\usepackage{enumerate}         
\usepackage{fixme}
\usepackage{color}
\usepackage{url}
\usepackage{amsthm}
\usepackage{bm}
\usepackage{xy}
\usepackage{enumitem}

\theoremstyle{plain}
\newtheorem{theorem}{Theorem}[section]

\newtheorem{corollary}[theorem]{Corollary}

\newtheorem{lemma}[theorem]{Lemma}

\newtheorem{definition}[theorem]{Definition}

\newtheorem{assumption}[theorem]{Assumption}
\theoremstyle{remark}
\newtheorem{remark}[theorem]{Remark}

\newtheorem{example}[theorem]{Example}

\numberwithin{equation}{section}

\newcommand{\ind}{1\!\kern-1pt \mathrm{I}}
\newcommand{\rsto}{]\!\kern-1.8pt ]}
\newcommand{\lsto}{[\!\kern-1.7pt [}

\numberwithin{equation}{section}

\renewcommand{\P}{\mathbb{P}}

\newcommand{\la}{\lambda}

\renewcommand{\rho}{\varrho}

\begin{document}
\title[FTAP in a two filtration setting]{A fundamental theorem of
  asset pricing for continuous time large financial markets in a two
  filtration setting}
\begin{abstract}
  We present a version of the fundamental theorem of asset pricing
  (FTAP) for continuous time large financial markets with two
  filtrations in an $L^p$-setting for $ 1 \leq p < \infty$. This
  extends the results of Yuri Kabanov and Christophe Stricker
  \cite{KS:06} to continuous time and to a large financial market
  setting, however, still preserving the simplicity of the discrete
  time setting. On the other hand it generalizes Stricker's
  $L^p$-version of FTAP \cite{S:90} towards a setting with two
  filtrations. We do neither assume that price processes are
  semi-martigales, (and it does not follow due to trading with respect
  to the \emph{smaller} filtration) nor that price processes have any
  path properties, neither any other particular property of the two
  filtrations in question, nor admissibility of portfolio wealth
  processes, but we rather go for a completely general (and realistic)
  result, where trading strategies are just predictable with respect
  to a smaller filtration than the one generated by the price
  processes. Applications range from modeling trading with delayed
  information, trading on different time grids, dealing with
  inaccurate price information, and randomization approaches to
  uncertainty.
\end{abstract}

\thanks{The authors gratefully acknowledges the support from ETH-foundation. Parts of this paper were written while the first author was visiting ETH Z\"urich; she is grateful to
the Forschungsinstitut f\"ur Mathematik.} 

\author{Christa Cuchiero, Irene Klein and Josef Teichmann}
\address{Vienna University, Oskar-Morgenstern-Platz 1, A-1090 Vienna
  and ETH Z\"urich, R\"amistrasse 101, CH-8092 Z\"urich}
\maketitle

\section{Introduction}\label{sec:intro}

One of the often unanimously accepted hypothesis in modeling financial
markets is the following:

\subsection*{Standard Hypothesis} \emph{Observations of prices are
  perfect and can be immediately included into trading decisions.}

\bigskip

\noindent It is the goal of this article to consider a setting beyond
this hypothesis: imagine a stock exchange with liquid prices given in
continuous time, whose informational content is encoded in a large
filtration $ \mathbb{G}$ and whose price processes are modeled by a
stochastic process $S$ adapted to this filtration. However, like in
Platon's famous allegory of the cave, the prices $S$ are not fully
revealed to us observers of the market, but only a shadow of them is
visible for us traders, i.e.~a perturbed observation of $S$. Nevertheless
we (have to) trade in the market given our observational basis. We
call this a \emph{platonic stock exchange}.

We believe that the platonic stock exchange helps to encode the
idealistic assumption of continuous time models and the actual
observational reality in terms of a combined filtering and trading
model. There are several instances, where the platonic stock exchange
can be directly applied:
\subsection*{:} Prices come on a discrete grid possibly with a certain
degree of reliability, hence the observational filtration is smaller
than the idealistic continuous time model filtration of the price
process.
\subsection*{:} Prices can additionally come with frictions
(transaction costs or liquidity), where the actually traded prices and
the observed prices, upon which trading decisions are based, do not
necessarily agree.
\subsection*{:} The setting of stochastic portfolio theory, where
relative capitalizations (prices relative to the market capitalization
taken as num\'eraire) are quoted on different orders of magnitude with
different precision and with different degrees of friction. Here
again, actually traded prices and observed prices do not necessarily
agree.
\subsection*{:} Market models jointly written for underlyings and
derivatives, where prices often come on different time grids,
e.g.~derivatives might be traded on a daily basis whereas underlyings
are usually traded on a much finer grid. It is useful to introduce a
difference between actually observed prices and traded prices.
\subsection*{:} A time delay in receiving market information (or in
applying it), which causes the trader's filtration simply being a
delayed one in comparison to a price filtration. For instance,
different time scales of trading for underlyings, which are traded on
a high frequency basis, occur in all major markets.
\subsection*{:} Markets with an execution delay of orders.
\subsection*{:} To quantify the effect of calibration errors, which
appear as additive error variables on market data. For instance in
term structure models market prices are only met approximately, which
is a consequence of (simple) inter- or extrapolation procedures.
\subsection*{:} ... and, of course, the setting of model uncertainty
combined with realistic market information as outlined in the sequel.
\bigskip

Formally speaking a \emph{platonic financial market} is given by a
stochastic basis $(\Omega,\mathcal{G},P) $ together with two
filtrations $ \mathbb{F} \subset \mathbb{G} $ and a family of
stochastic processes $S$ (without any assumption on path properties)
adapted to $\mathbb{G}$. Trading in the given assets is possible but
only with $\mathbb{F}$-predictable, simple strategies (and limits of
such strategies, which is made precise later). The smaller filtration
$ \mathbb{F} $ corresponds to the agent's information which actually
enters into trading decisions, whereas the larger filtration
$\mathbb{G} $ encodes all information from price processes which is
not necessarily instantaneously available for trading decisions. The precise setting that we introduce in Section 2 is even more general in terms of the involved filtrations, but for the sake of simplicity we do not enter into details here. We
emphasize (and outline this more precisely in the sequel) that this framework is
also a way to incorporate model \emph{uncertainty}, since a lack of
observations actually leads to uncertainty in the model choice
itself. In any case the measure $P$ does not necessarily have the
meaning of a historical measure: one more natural interpretation is a
randomization by a subjective prior of a class of models among which
one cannot distinguish by actual observations.

As a novelty here we do neither assume that the price processes are
semi-martingales (and there is also no reason to do so since trading
is only using strategies predictable with respect to the smaller
filtration $\mathbb{F}$), nor do we pose a standard filtering problem
by filtering a ``true price". We consider the non-semimartingality
(and the absence of path properties) of price processes as a
particular challenge as well as an advantage of our approach, for
instance in view of explaining well known econometric evidence, as we
encounter it for instance in high frequency data. 

The main goal of the
present article is to investigate all foundational questions of
Mathematical Finance, i.e.~fundamental theorems, superhedging and
duality in this new setting of two filtrations.
Our main result states that a certain ``No-arbitrage'' condition (see
Definition \ref{CpNAFLp}) for \emph{large} financial markets in a two filtration setup, which
we call (NAFLp), is equivalent to the existence of an equivalent
measure under which the optional projections of price processes $S$ on
the smaller filtration $\mathbb{F}$ are martingales (see Theorem
\ref{ftap-p} and Corollary \ref{ftap-1}). Our ``No-arbitrage'' condition
involves $L^p$-integrability and $L^p$-convergence of terminal portfolio
values with respect to some measures $P'$, for $ 1 \leq p < \infty $,
and turns out to be equivalent to the classical (NFLVR) condition in the case of one filtration and bounded price 
processes. The concept is appealing since topologies are actually
quite strong, economically reasonable and no further admissibility
assumptions are needed.

To the best of our knowledge this is the first \emph{fundamental
  theorem of asset pricing} (FTAP) in continuous time when trading
only with respect to a smaller filtration is considered. It thus
extends the results of Y.~Kabanov and C.~Stricker \cite{KS:06} to
continuous time and to a large financial market setting. On the other
hand it generalizes Stricker's $L^p$-version of FTAP \cite{S:90}
towards a setting with two filtrations. This $L^p$-setting allows to
tranfer the simplicity of discrete time Mathematical Finance to
continuous time, i.e.~no assumptions on paths and no assumptions on
admissibility. Our setup also allows to overcome the disadvantage of
Stricker's setting that the no-arbitrage condition depends on the
measure and not only on its equivalence class, however, there are
different ways to do so: we have therefore provided several competing
(and equivalent) versions of our ``No-arbitrage'' condition. The
obtained two filtration FTAP then constitutes the basis for
superhedging results, where trading is again only allowed with respect
to the smaller filtration. As shown in Example \ref{fin}, our setting
naturally embeds semi-static hedging. Combined with \emph{Bayesian
  uncertainty modeling} as explained below these superhedging results
give a new flavor to the corresponding results in the area of robust
finance (see e.g. the work of Bruno Bouchard and Marcel Nutz~\cite{BouchardNutz2015}).

The remainder of the article is organized as follows. The following
Subsection \ref{sec:Bayes} introduces the concept of \emph{Bayesian
  uncertainty modeling}, while Subsection \ref{sec:Examples} underpins
the relevance of our approach by means of examples. In
Section~\ref{sec:setting} we introduce the formal setting of platonic
large financial markets, while Section~\ref{main} and
\ref{sec:super-replication} are dedicated to ``No-arbitrage''
definitions, FTAP and superheding results.

\subsection{Bayesian uncertainty modeling} \label{sec:Bayes}

The above described two filtration setting is the first step towards a \emph{dynamic}
  framework for modeling uncertainty where a real-time inclusion of
new information and thus a decrease (or increase) in model uncertainty
can be analyzed. We understand uncertainty here from a Bayesian
viewpoint, i.e.~we randomize over different measures \footnote{Josef
  Teichmann is grateful to Dima Kramkov for many discussions during
  morning runs about this viewpoint.}.  We refer to this as
\emph{Bayesian uncertainty modeling} and call the whole approach \emph{Bayesian Finance},
since optional projections will play a key role such as in Bayesian
Filtering.

Let us outline what we mean by \emph{Bayesian uncertainty
  modeling} and how this is naturally linked to a two-filtration
framework: assume here some path space $D$ together with its canonical
filtration $ \widetilde{\mathbb{F}}$ and a family of (canonical) price
processes $ S $, adapted to $\widetilde{\mathbb{F}}$. Furthermore we
are given a family of probability measures $ {P}^\theta $ for a
parameter $ \theta \in \Theta $. We assume an a priori given
probability measure $ \nu $ on $ \Theta $ (of course we could at this
point also introduce some time dependence on $\Theta$ but we leave
this away for the sake of simplicity). Formally, our two filtration
setup can be introduced in the following way: consider first on
$ D \times \Theta $ the (full information) filtration
$ \mathbb{G} := \widetilde{\mathbb{F}} \otimes \mathcal{B}(\Theta) $,
and the probability measure
\[
  \P(A \times B) = \int_B {P}^\theta(A) \nu(d \theta) \, .
\]
Of course price processes can also be defined on $ D \times \Theta $,
namely simply via $S_t(\omega,\theta)=\omega(t)$. In the case of
perfect information we find ourselves again in a classical setting as
far as trading is concerned, since the trader could use the
information generated by $ S $.

We now encode the very nature of market data, i.e.~the basis for
trading and for calibration, which usually come
\begin{itemize}
\item on a discrete, not necessarily equidistant and possibly
  unpredictable time grid;
\item with an additional degree of non-reliability;
\item with an idea on acceptable calibration accuracy (i.e.~liquid
  prices should be perfectly calibrated, less liquid ones less
  perfectly, etc);
\end{itemize}
by specifying a smaller filtration
$ \mathcal{F}_t \subset \sigma\big({(S_s)}_{0 \leq s \leq t}\big)
\subset \widetilde{\mathcal{F}}_t \otimes \mathcal{B}(\Theta) $, for
$ t \geq 0 $, on the product probability space, which corresponds to
the information of the actual observations being often strictly
smaller than the filtration generated by $S$. Even if data are fully
reliable, already the discrete character of observations causes a
filtration shrinkage. If data are additionally not fully reliable (due
to market frictions or simple observational issues or delays), the
smaller filtration corresponds to the actual observation whereas the
larger filtrations account for model uncertainty and additionally for
some noise perturbing, e.g.~additively, the market data. Now we find
ourselves in a full-fledged two-filtration setting (actually a three
filtration setting: full information on price process $S$ and
$\theta$, only full information on $S$ and actual observation
information on $S$), where trading strategies are predictable with
respect to a smaller filtration $ \mathbb{F} $ than the one generated
by the actual (idealistic) price processes $S$, and where full
knowledge on price processes and perturbing noise is encoded in the
largest filtration $ \mathbb{G} $. Hence we can analyze in this
setting how market observations lead to less uncertainty on the one
hand, i.e.~updates of $\nu$, and how market observations are used for
trading on the other hand.

This  Bayesian uncertainty setting can be simplified to a 
\emph{mixture model} setting where we work in contrast to the above one not on
the product space $D \times \Theta$ but on $D$ itself and consider a
probability measure of the form
\[
  \P(A) = \int_\Theta {P}^\theta(A) \nu(d \theta).
\]
While this setting already allows to give answers to many questions coming up in robust finance (see Example \ref{fin}), it can for instance not include derivatives whose payoff depends on the parameter $\theta$ (see Example \ref{ex:thetadepence}). Note that not only the name of this setup but the whole modeling is in spirit of mixture models as considered by Damiano Brigo and Fabio Mercurio \cite{BM:02}.

This Bayesian approach to uncertainty encodes subjective believes on
$ \Theta $ as probabilities, which are by no means probabilities in
the sense of risk with respect to the physical measure. However, in
finance we are anyway accustomed to probabilities, which do not
necessarily have the meaning of statistical risk, namely, risk-neutral
probabilities. Here we add a third dimension what a probability
actually means, namely a subjective belief in the validity of a model.

Of course, we are not the first ones to introduce Bayesian viewpoints
in Finance, see, e.g. the pioneering paper by Per Mykland \cite{mykland2003} or the
large literature on dynamic risk measures (e.g., the work by Beatrice Acciaio and Irina Penner
\cite{AcciaioPenner11} and the references therein) where updating and
time consistency plays a major role. In this respect we refer in
particular to \cite{CherKu11, Riedel04, BaeuerleMundt08}.  However, up
to our knowledge we are the first ones to combine Bayesian methods
consequently with FTAP and duality concepts. In other words, a long
term goal is a combination of the two seminal works \cite{KS:06} and \cite{rog:14}.

\subsection{Examples}\label{sec:Examples}

The following examples illustrate the universal applicability of our
two-filtrations setting and its strength in combination with Bayesian
uncertainty modeling.

\begin{example} [\textbf{Finitely many assets with a large
family of semi-statically traded options}]\label{fin}
Suppose we are in the setting of finitely many liquidly traded assets,
i.e.~we consider a finite number of adapted processes $S^1,\dots S^n$
adapted to a filtration $ \mathbb{F} $, together with a (large) family
of option price processes $\pi(f^j)$ paying off at time $t=1$ the
payoff $f^j$, for $ j \in J $ with $ J $ some index set. Let us be
more precise here: the process $ {(\pi_t(f^j))}_{t \in I^j}$ comes
with a time grid $ I^j \subset [0,1] $ being a finite union of
disjoint closed intervals, where the option is actually traded. A
particular example is $I^j=\{t^j_1,\ldots, t^j_n\}$ meaning that the
option $j$ can be only traded at times $t^j_k$,
$k \in \{1, \ldots, n\}$.

Additionally we apply the mixture setting of Bayesian uncertainty as
in the last part of Subsection~\ref{sec:Bayes}: consider a family of
probability measures $ {P}^\theta $ for a parameter
$ \theta \in \Theta $, where $ \Theta $ carries an \emph{a priori
  measure} $ \nu $. Assume that
\[
  \P(A) = \int_\Theta {P}^\theta(A) \nu(d \theta)
\]
is well-defined. Notice that nullsets with respect to this mixture
measure are nullsets with respect to $\nu$-almost every $P^\theta$.

The filtration $ \mathbb{G} $ is just chosen to be the constant
filtration $ \mathcal{G}_t:=\mathcal{F}_1 $, for $ t \in
[0,1]$. Assume that prices $S^i_t$ on
$(\Omega,\mathcal{G}, (\mathcal{G}_t)_{t\in[0,1]},P)$ lie in some
$L^p(\Omega)$ with trading based on $\mathbb{F}$. With respect to this
filtration we can of course extend the price to a $\mathbb{G}$-adapted
process on $ [0,1] $ via
\[
  \pi_t(f^j) := \pi_{\min \{ s \geq t \, | \; s \in I^j\}}(f^j) \, .
\]
This process is \emph{not} adapted anymore to the filtration
$ \mathbb{F} $ but rather to $ \mathbb{G} $. Since we do not require
path properties it does not matter that actually this process is
c\`agl\`ad.

The advantage of this construction is that semi-static hedging on
$ I^j $ is now expressible just via the standard stochastic
integral\footnote{We are grateful to Matteo Burzoni and Martin Larsson
  for pointing this out.}. Here $ \mathbb{G} $ needs to be a
filtration as large to contain already information on the values of
the options at the next future trading
time.
Since we do not have any assumptions on $\mathbb{G}$ this can of
course be assumed and our chosen filtration clearly does the job.

The super-replication result from Section \ref{sec:super-replication}
then reads as follows: for every $f \in L^p(\P)$
$$\sup_{Q\in\mathbb{M}^{q}}E_Q[f]=\inf\{x: \exists \,
g\in \overline{C_{p}}\text{ with } x+g \geq f\} \, .$$ The set
$ \mathbb{M}^q $ consists of measures $ Q \sim \P $,
$\frac{dQ}{d\P} \in L^q(\P) $ such that
\[
  E_Q[\pi_t(f^j)\, | \, \mathcal{F}_s]=\pi_s(f^j)
\]
for $ s \leq t $ in $I^j$ (sic!) and for all $ j \in J $, and, of
course, such that the optional projection (which conincides with $S$
itself) of every price process $S^i$ on $\mathbb{F}$ is a
$ Q$-martingale. On the other hand the super-hedge is understood
$\P$-almost surely, i.e.~for each $f \in L^p(\P)$ there are sequences
of simple trades in $ S $ and semi-static trades in finitely many
options minus some consumptions converging in $L^p(\P)$ to a limit $g$
which dominates $f$ minus the super-hedging price $x$ $\P$-a.s.,
i.e. on a measurable set $A\subset\Omega$ satisfying
$ P^\theta (A) = 1 $ for $\nu$-almost all $\theta$ (for details see
Section~\ref{sec:super-replication}).  This is a first super-hedging
result in continuous time with a large family of options traded on
different time grids in a robust setting interpreted in a Bayesian
way.
\end{example}

\begin{remark}
  We could of course have considered a full-fledged Bayesian setting
  of robustness where also the price processes are not adapted to the
  filtration $ \mathbb{F} $.
\end{remark}

\begin{remark}
  The trick to include semi-static hedges by writing piece-wise
  constant processes with anticipating information (depending on the
  amount of static properties the hedging should have: if one wants to
  have static hedging on the interval $ ]s,t] \subset [0,1] $, then
  time $t$ information must already be present at $s+$ which due to
  our generality of $ \mathbb{G}$ is feasible) works in general and
  one can therefore model \emph{trading on different grids for each
    single asset} within a two filtration setting. This makes our
  setting extremely general in scope. This is in line with discrete
  time small financial market setting in \cite{BouchardNutz2015}.
\end{remark}

\begin{example} [\textbf{finitely many assets with a large family of
    semi-statically traded options and uncertainty swaps}]\label{ex:thetadepence}

  Assume the setting of Bayesian uncertainty, as of Section
  \ref{sec:Bayes}, which allows to include artificial derivatives
  $ \pi(f^j) $, for $ j \in J $ being traded at time $t=0$ at an
  $\mathcal{F}_0$-measurable price.  We imagine here payoffs $ f^j $
  which depend on the uncertainty parameter $ \theta $ and call them
  \emph{uncertainty swaps} (even though they do not necessarily have a
  price $ 0 $ at time $ t=0$). These swaps represent risks related to
  uncertainty, i.e.~how likely certain areas of $ \Theta $ are. One
  can think of uncertain volatility, i.e.~the parameter $ \theta $
  represents all possible volatility configurations in the market, or
  of risks not fully reflected in price behavior like, e.g.,
  temperature in energy markets.
\end{example}

\begin{example} [\textbf{An asset with uncertain price}]

  Probably the simplest and most extreme example of a two filtration
  setting is the following: assume a standard price process $Y$ on a
  filtered probability space $(\Omega,\mathbb{G},P)$, for instance a
  Black-Scholes or Heston price model, and assume additionally the
  existence of a centered (i.e.~expectations with respect to $P$ are
  vanishing), uniformly bounded $\mathbb{G}$-adapted process $ Z $
  fully independent of $ Y $ whose values at each time are
  \emph{independent} of all other values at any time, as it is often
  assumed when modeling micro structure noise. Of course $Z$ cannot
  have any reasonable path properties. Define $ S:=Y+Z$, then the
  price looks like $Y$ when observed over some time interval, but
  trading might lead to surprises. Our setting allows to derive
  super-replication results in such extreme but realistic cases.
\end{example}

\section{Large Platonic financial markets}\label{sec:setting}

We consider a large platonic financial market model in continuous time
in the following way. Let $I$ be an arbitrary parameter space which
can be any set, countable or uncountable. Let $T=1$ denote the time
horizon and let $(\Omega,\mathcal{G},P)$ be a probability space with a
filtration $\mathbb{G}=(\mathcal{G}_t)_{t\in[0,1]}$. On this
probability space we are given a family of $\mathbb{G}$-adapted
stochastic processes $(S^i_t)_{t\in[0,1]}$, $i\in I$. In particular no
path properties are needed. $P$-almost surely is understood with
respect to $ \mathcal{G} $ being the largest $ \sigma $-algebra in
this setting.

We define, for each $n\geq 1$, a family $\mathcal{A}^n$ of subsets of
$I$, which contain exactly $n$ elements:
\begin{equation}\label{eq:n-sets}
  \mathcal{A}^n=\{\text{all/some subsets }A\subseteq I\text{, such that $|A|=n$}\},
\end{equation}
where $|A|$ denotes the cardinality of the set $A$. Moreover, we
assume that if $A^1,A^2\in \bigcup_{n\geq 1}\mathcal{A}^n$, then
$A^1\cup A^2\in \bigcup_{n\geq 1}\mathcal{A}^n$ (\emph{refining
  property}).

We consider a family of filtrations
$\mathbb{F}^A=(\mathcal{F}^A_t)_{t\in[0,1]}$, indexed by
$A\in\bigcup_{n\geq 1}\mathcal{A}^n$, which are all contained in (and
possibly smaller than) the filtration
$\mathbb{G}=(\mathcal{G}_t)_{t\in[0,1]}$ introduced
above. Additionally we suppose for two sets
$A^1, A^2\in\bigcup_{n\geq 1}\mathcal{A}^n$, such that
$A^1\subseteq A^2$, that $\mathbb{F}^{A^1}\subseteq \mathbb{F}^{A^2}$,
i.e., for each $t$,
$\mathcal{F}^{A^1}_t\subset \mathcal{F}^{A^2}_t\subseteq
\mathcal{G}_t$ (\emph{monotonicity property}).

For each $A\in\bigcup_{n\geq 1}\mathcal{A}^n$ we define the following
set of portfolio wealth processes based on simple strategies for
deterministic time points in the \emph{small financial market $A$}
that are predictable with respect to the smaller filtration
$\mathbb{F}^A=(\mathcal{F}^A_t)_{t\in [0,1]}$. To be precise:
\begin{definition}\label{XA}
  Let $t_0,\dots,t_l\in[0,1]$ denote a set of deterministic time
  points and consider a small market indexed by
  $A=\{\alpha_1,\dots,\alpha_n\}\in \mathcal{A}^n$. Denote
  $\mathbb{F}^A$-simple and bounded processes by
  \[
    \mathbf{H}^A=\sum_{i=1}^{l}H^{A}_{t_{i-1}}1_{(t_{i-1}, t_i]}
  \]
  with
  $\mathbf{H}^{A}_{t_{i-1}}=(H^{\alpha_1}_{t_{i-1}}, \ldots,
  H^{\alpha_n}_{t_{i-1}})^{\top}\in \mathcal{F}^{A}_{t_{i-1}}$ for all
  $i \in \{1, \ldots, l\}$.  Then the set of simple portfolio wealth
  processes obtained from bounded, $\mathbb{F}^{A}$-simple trading
  strategies is defined as
$$\mathcal{X}^A=\{(\mathbf{H}^A\cdot \mathbf{S}^A)_{t \in [0,1]}: \mathbf{H}^A\text{ $\mathbb{R}^n$-valued, bounded, $\mathbb{F}^A$-simple}\},
$$
where $\mathbf{S}^{A}=(S^{\alpha_1}, \ldots, S^{\alpha_n})^{\top}$ and
$$(\mathbf{H}^A\cdot \mathbf{S}^A)_{t}=
\sum_{j=1}^n\sum_{i=1}^lH^{\alpha_j}_{t_{i-1}\wedge
  t}(S^{\alpha_j}_{t_i\wedge t}-S^{\alpha_j}_{t_{i-1}\wedge t}),$$
meaning that trading is done in an $\mathbb{F}^A$-predictable way.
\end{definition}

Next we define the set $\mathcal{X}^{n}$ of all portfolio wealth
processes with respect to simple strategies that include at most $n$
assets (but all possible different choices of $n$ assets). Indeed, for
each $n\geq 1$, we consider the following set $\mathcal{X}^n$
\begin{equation}\mathcal{X}^{n}=\bigcup_{A\in\mathcal{A}^n}\mathcal{X}^{A}.
\end{equation}
Note that the sets $\mathcal{X}^{n}$ are neither convex nor do they
satisfy a concatenation property in the sense of \cite{kab:97},
because in both cases $2n$ assets could be involved in the
combinations. Therefore the result would rather be in the larger set
$\mathcal{X}^{2n}$ than in $\mathcal{X}^{n}$.

\begin{definition}\label{largeX1}
  We introduce the convex sets of ($\mathbb{F}$-simple) portfolio
  wealth processes, its terminal evaluation and the convex cone of all
  super-replicable claims:
  \begin{enumerate}
  \item Define the set of all wealth processes defined on
    $\mathbb{F}$-simple strategies involving a finite number of assets
    in the large financial market as
    $\mathcal{X}=\bigcup_{n\geq 1}\mathcal{X}^n$.
  \item We denote by $K_0=\{X_1: X\in\mathcal{X}\}$ the evaluations of
    elements of $\mathcal{X}$ at terminal time $T=1$.
  \item We denote by $C$ the convex cone of all super-replicable
    claims (by $\mathbb{F}$-simple strategies) in the large financial
    market, that is,
$$C=K_0-L_+^0(\Omega, \mathcal{G},P).$$
\end{enumerate}
\end{definition}

\begin{remark}
  So far our setting is not only completely general but also
  extraordinarily realistic in the sense that it can fully capture all
  desired features mentioned in the introduction, in particular
  trading with delay and market frictions. Note that we do not need to
  assume any path properties for price processes.
\end{remark}

Note that the above setting includes as examples the large financial
market based on a sequence of assets as in the work of Marzia DeDonno, Paolo Guasoni and Maurizio Pratelli \cite{DonnoGuasoniPratelli05} as well as bond markets (with a
continuum of assets), with trading as specified in Definition
\ref{XA}.  For a more detailed discussion see
\cite{CuchieroKleinTeichmann}.

\section{No asymptotic $L^p$-free lunch and FTAP}\label{main}

For $1\leq p<\infty$ we denote in the sequel
$L^p(\Omega,\mathcal{G},P)=L^p(\mathcal{G},P)$. Moreover, for some set
$E \subset L^p(\mathcal{G},P)$ we denote by
$\overline{E}=\overline{E}^{\| \cdot\|_{L^p(\mathcal{G},P)}}$ the
$L^p(\mathcal{G},P)$-closure of $E$.

The crucial assumption which allows us to work in an $L^p$-setting, as Stricker did in his work~\cite{S:90} in the setting of one filtration and small markets, is:

\begin{assumption}\label{lp}
  For some fixed $p$, $1\leq p<\infty$, we denote by $\mathcal{P}_p$
  the following set of measures
  \[
    \mathcal{P}_p =\{ P' \sim P \, | \, S^i_t\in L^p(\mathcal{G},P')
    \text{ for all } i\in I, t\in[0,1] \}.
  \]
  We assume that $ \mathcal{P}_p \neq \emptyset $, i.e. that there is
  an equivalent probability measure $P' \sim P$ such that
  $S^i_t\in L^p(\mathcal{G},P')$, for all $i\in I$, $t\in[0,1]$.
\end{assumption}

Note that in the case of countably many assets, i.e., when $I$ is
countable, Assumption~\ref{lp} is always satisfied (for each $p$).

\begin{remark}\label{weakerass}
  Assumption~\ref{lp} can be slightly weakened. It is enough to assume
  that there exists some $P' \sim P$ such that
  $(S^i_u-S^i_t)^-\in L^p(\mathcal{G},P')$, for all $i\in I$,
  $t\leq u\in[0,1]$ for some fixed $1\leq p<\infty$ when we only
  consider long-only investments in the assets. The corresponding
  result will then be slightly weaker in the sense that we will only
  get a measure such that the optional projections are
  supermartingales (and not martingales as in the case of
  Assumption~\ref{lp}). We refer to Section \ref{main_weakened} for
  the corresponding result.
\end{remark}

We can now define a notion of absence of arbitrage, notably without
applying stochastic integration, which -- at this point -- would not
be available in full generality.

\begin{definition}\label{CpNAFLp}
  We say that the large financial market satisfies the condition {\it
    no asymptotic $L^p$-free lunch} (NAFLp) if there is a probability
  measure $P'\sim P$ as in Assumption~\ref{lp} such that the following
  holds:
  \begin{align}\label{NA}
    \overline{C_p(P')}\cap L^p_+(\mathcal{G},P')=\{0\},
  \end{align}
  where $C_p(P')=C\cap L^p(\mathcal{G},P')$ with $C$ introduced in
  Definition \ref{largeX1} (iii).
\end{definition}

From \cite{KS:06} we know that $C$ is closed in $L^0$ when dealing with finitely many assets and finitely many trading times. Hence, elements in $\overline{C_p(P')}$  which do not lie in $C_p(P')$ necessarily involve infinitely many assets and \slash or infinitely many trading times.

\begin{remark}\label{cp=k0minuslp}
  \begin{enumerate}
  \item It is obvious that $C_p(P')\neq \emptyset$, since strategies
    are bounded. Indeed, $K_0\subseteq L^p(\mathcal{G},P')$ and
    therefore
  $$C_p(P')=K_0-L^p_+(\mathcal{G},P').$$
\item It would also be possible to consider the set
  $\mathcal{C}_p:=C\cap \bigcap_{P' \in \mathcal{P}_p}
  L^p(\mathcal{G},P')$. Since
  $K_0\subseteq \bigcap_{P' \in \mathcal{P}_p}L^p(\mathcal{G},P')$, we
  have similarly as above
  \[
    \mathcal{C}_p:=K_0 - \bigcap_{P' \in \mathcal{P}_p}
    L_+^p(\mathcal{G},P').
  \]
  For some fixed $P'$, it then holds that the
  $L^p(\mathcal{G},P')$-closures of $C_p(P')$ and $\mathcal{C}_p$ are
  the same. Indeed if $g$ in the closure is the $L^p(P')$-limit of
  $g_n=f_n-h_n$ where $f_n\in K_0$ and $h_n\geq0$ we can always choose
  $h_n\in L^{\infty}_+(\mathcal{G})\subseteq \bigcap_{\widetilde{P}
    \in \mathcal{P}_p} L_+^p(\mathcal{G},\widetilde{P})$ as
  $L^{\infty}_+(\mathcal{G})$ is dense in $L^{p}_+(\mathcal{G}, P')$
  for the $L^p(P')$-norm. 

Furthermore, \eqref{NA} is equivalent to
  \begin{align}\label{alt_intersec}
    \overline{\mathcal{C}_p}^{L^p(\mathcal{G}, P')} \cap \bigcap_{\widetilde{P} \in \mathcal{P}_p}L_+^p(\mathcal{G},\widetilde{P})=\{0\} \, .
  \end{align}
  Indeed, suppose \eqref{alt_intersec} holds but \eqref{NA} does not hold. By the
  above there is  $g\in \overline{\mathcal{C}_p}^{L^p(\mathcal{G}, P')}$ with
  $g\geq 0$, $g\neq0$, which is the $L^p(P')$-limit of
  $g_n=f_n-h_n\in \bigcap_{\widetilde{P} \in
    \mathcal{P}_p}L^p(\mathcal{G},\widetilde{P})$. Then, clearly
  $\widetilde{g_n}=g_n-(g_n-1)\ind_{\{g_n\geq1\}}\in
  \bigcap_{\widetilde{P} \in
    \mathcal{P}_p}L^p(\mathcal{G},\widetilde{P})$ as well and
  converges in $L^p(\mathcal{G},P')$ to
  $g\land 1$ which lies in  $\cap_{\widetilde{P} \in
    \mathcal{P}_p}L^p_+(\mathcal{G},\widetilde{P})\setminus\{0\}$,
  yielding thus a contradiction. The other direction is clear.
\end{enumerate}
\end{remark}

The following example illustrates that the choice of
$\overline{C_p(P')}$ instead of $\overline{K_0}$ in the definition of
(NAFLp) is crucial beyond the setting of small financial markets in
discrete time.

\begin{example}
  A careful reading of Example~3.3 of \cite{S:94} shows that it is not
  possible to replace $\overline{C_p(P')}$ by $\overline{K_0}$ in the
  definition of (NAFLp). Indeed, let $p=1$ and consider a one period
  market with countably many derivatives given at time $t=1$ by the
  random variables $ f_n$ of the Example~3.3 in \cite{S:94} and at
  price $0$ at time $ t=0$. Again, we can create, as in the
  introduction, a two filtration setting, where hedging is actually
  buy \& hold in this large financial market, and the filtration
  $ \mathbb{F} $ is trivial. In this setting $ K_0 $ contains all
  (finite) linear combinations of $ f_n $. As in \cite{S:94} we can
  show that $g_n=\sum_{k=1}^nf_k\in K_0$ is bounded below by -1, for
  each $n$, and $P'(g_n\geq 1)\to1$. Hence
  $\tilde{g_n}=g_n-(g_n-1)\ind_{\{g_n\geq 1\}}\to 1$ in $L^1$ by
  dominated convergence. Therefore $1\in \overline{C_1(P')}$ and
  \eqref{NA} is not satisfied. However, analogously as in \cite{S:94}
  we can show that $\overline{K_0}\cap L^1_+(\mathcal{G}, P')=\{0\}$.
\end{example}

\begin{remark}
  We emphasize that we do not assume any admissibility for our
  portfolio wealth processes, instead we assume $L^p$-integrability
  with respect to a measure $P'$ equivalent to the physical measure
  $P$. This follows the setting of \cite{S:90}. However it does
  \emph{not} share the disadvantage that the respective no arbitrage
  condition depends on $P$ itself, but only on the equivalence class
  of $P$ which is a desirable feature and in particular the case for
  the classical NFLVR condition introduced by Freddy Delbaen and Walter Schachermayer in 
  \cite{DS:94}. Furthermore we do not need a stochastic integration
  theory at hand, which, in our general setting and in contrast to the
  setting of \cite{S:90} is not (yet) available.
\end{remark}

Fix now $ 1< q \leq \infty $ dual to $p$, i.e.
$\frac{1}{p} +\frac{1}{q}=1$, for $p$ given in Assumption \ref{lp}.
We define the set of $L^q$-probability measures for which the optional
projection of the process $(\mathbf{S}^A_t)$ with respect to the
filtration $\mathbb{F}^A$ is a martingale, for all finite subsets $A$
of $I$, as follows:
\begin{equation*}
  \begin{split}
    \mathcal{M}^q =\{&Q\sim P \, |\, \exists \, P' \in \mathcal{P}_p
    \text{ s.t. } \frac{dQ}{dP'}\in L^q(\mathcal{G},P') \text{ and
    }E_Q[S^{\alpha_i}_u|\mathcal{F}^A_t]=
    E_Q[S^{\alpha_i}_t|\mathcal{F}^A_t]\text{ a.s.,}\\ &\text{for all
    }A=\{\alpha_1, \ldots, \alpha_l\} \in\bigcup_{n\geq
      1}\mathcal{A}^n, 1\leq i\leq l\text{ and all }t\leq u\in[0,1]\}.
  \end{split}
\end{equation*}

Moreover, for $q=1$, we define the analogous set of equivalent
probability measures without additional property on the $q$th moments
of the Radon Nikodym density, i.e.,
\begin{align*}
  \mathcal{M}^1 =\{&Q\sim P \,|\, E_Q[S^{\alpha_i}_u|\mathcal{F}^A_t]=
                     E_Q[S_t^{\alpha_i}|\mathcal{F}^A_t]\nonumber\text{ a.s.,}\nonumber\\ &\text{for all }A=\{\alpha_1,\dots,\alpha_l\}\in\bigcup_{n\geq 1}\mathcal{A}^n, 1\leq i\leq l\text{ and all }t\leq u\in[0,1]\}
\end{align*}

Let us remark that the only instance where the filtrations $\mathbb{F}^A$ introduced in Section 2 actually occur explicitly is in the above definition  of the dual objects $\mathcal{M}^q$ for $ 1 \leq q \leq \infty $.

\begin{remark}\label{convexity}
  Clearly, $ \mathcal{M}^1 $ is convex. For $ 1 < q \leq \infty $, the
  sets $ \mathcal{M}^q $ are convex as well. Indeed, for this purpose
  we consider the following slightly more general statement: For all
  $ Q_i \ll P $ with $ \frac{dQ_i}{dP_i} \in L^q(\mathcal{G}, P_i) $
  for measures $ P_i \sim P $, $ i\in \{1,2\}$, the convex
  combinations $ Q_s:= sQ_2 + (1-s)Q_1 $ satisfy
  $ \frac{dQ_s}{d \widetilde P} \in L^q(\mathcal{G},\widetilde P) $
  for $ 0 \leq s \leq 1 $, where
  $ \widetilde P = \frac{1}{2} (P_1 + P_2) $. Suppose
  $ 1 < q < \infty $, since the assertion is clear for $ q = \infty
  $. Indeed we have
  \[
    E_{\widetilde P}\bigg[ {\left|\frac{dQ_i}{d \widetilde
          P}\right|}^q \bigg] \leq 2^{q-1} E_{\widetilde P}\bigg[
    {\left|\frac{dQ_i}{d \widetilde P}\right|}^q {\left|\frac{d
          \widetilde P}{d P_i}\right|}^{q-1} \bigg] = 2^{q-1}
    E_{P_i}\bigg[ {\left|\frac{dQ_i}{d P_i}\right|}^q \bigg] < \infty
  \]
  for $ i = 0,1 $, since $ 2 \frac{d\widetilde P}{dP_i} \geq 1 $. The
  rest follows by the triangle inequality.
\end{remark}

\begin{remark}\label{stronger_topologies}
  The above convexity assertion is related to the fact that locally
  convex vector spaces formed as intersection of spaces
  $ L^p(\mathcal{G}, P') $, where $ P' \sim P $ runs over a set of
  probability measures subject to additional constraints (e.g., as in
  our case probability measures $P'$ such that all price processes are
  $p$-integrable), has as a (strong) dual space the union of
  $L^q(\mathcal{G},P')$ with respect to the same family of measures
  $P'$. The corresponding topologies are the projective and injective
  locally convex topologies, i.e.~the initial and final topologies
  making all canonical maps continuous. Let us formulate this more
  directly in case of $X:=\cap_{P' \in \mathcal{P}_p} L^p(P') $: note
  that $ \mathcal{P}_p$ is a directed set inheriting its (reflexive,
  transitive and anti-symmetric) relation ``$\leq$'' from reversing
  the inclusion of the spaces $L^p(\mathcal{G},P')$.  Indeed,
  $P' \leq \hat{P}$, if
  $L^p(\mathcal{G}, \hat{P}) \subseteq L^p(\mathcal{G},P')$.  For
  $P' \leq \hat{P}$, consider the inclusions from
  $L^p(\mathcal{G},\hat{P}) \to L^p(\mathcal{G},P')$. Then, $X$ is the
  projective limit with respect to these mappings. The topology of $X$
  is now the coarsest topology on $X$ which makes the inclusion maps
  from $X$ to $L^p(\mathcal{G}, P')$ continuous.  Therefore,
  intersections of $ L^p(\mathcal{G}, P') $ balls of some radius
  around $0$ with $ X $ constitute a neighborhood base for this
  topology. Hence any linear functional $\ell$ with respect to this
  topology can be extended some $L^p(\mathcal{G}, \hat{P})$ with
  $\hat{P}$ in $\mathcal{P}_p$, just by the fact that the open
  neighborhood $ \ell^{-1}(]-1,1[) $ of $0$ has to contain some
  intersection of $L(\mathcal{G}, P')$ balls 
  with $X$. Therefore, $\ell$ can be represented as
  \[
    \ell: X \to \mathbb{R},\, f\mapsto \int f g d\hat{P}
  \]
  for some $g \in L^p(\mathcal{G}, \hat{P})$.  Combining this with the
  fact that all linear functionals of this form for some
  $\hat{P} \in \mathcal{P}_p$ and some
  $g \in L^q(\mathcal{G}, \hat{P})$ are elements of the dual, yields
  the assertion that the strong dual of $X$ is actually
  \[
    \bigcup_{P' \in \mathcal{P}} L^q(\mathcal{G}, P').
  \]
  The strong topology on the strong dual just means that convergence
  always takes place in some $L^q(\mathcal{G}, P')$.

  In this direction one could also work with the locally convex vector
  space
  \[
    \bigcap_{p \geq 1} \bigcap_{P' \in \mathcal{P}_p} L^p(\mathcal{G},
    P')
  \]
  and its natural projective limit topology. The corresponding ``No
  arbitrage'' condition would be to replace the fixed $p$ in
  Definition \ref{CpNAFLp} by `` there exists some $p$'' and
  corresponding set of martingale measures would then be
  $\cup_{q > 1} \mathcal{M}^q$ .
 
\end{remark}

We have now collected all ingredients to formulate a \emph{fundamental
  theorem of asset pricing} in the present context of two filtrations.

\begin{theorem}\label{ftap-p}
  Suppose that Assumption \ref{lp} holds for some fixed
  $1 \leq p< \infty$.  Then the condition (NAFLp) holds if and only if
  $\mathcal{M}^q\neq\emptyset$ where $q$ satisfies
  $\frac1{p}+\frac1{q}=1$.
\end{theorem}

\begin{proof}
  Assume first that (NAFLp) holds. This means that there exists some
  $P' \sim P$ such that \eqref{NA} holds.  Note that this then implies
  Condition (ii) of Theorem \ref{th:asy} of the Appendix for
  $M=C_p(P')$. This in turn is equivalent to Condition (iii) in
  Theorem \ref{th:asy} and thus yields some
  $Z \in L^q(\mathcal{G}, P')$ such that $Z >0$ a.s.~and
  $\sup_{f \in \overline{C_p(P')}} E_{P'}[Zf]< \infty$. As $C_p(P')$
  is a convex cone this implies that
  $\sup_{f \in \overline{C_p(P')}} E_{P'}[Zf]\leq0$. Define now $Q$
  with $\frac{dQ}{dP'}=\frac{Z}{E_{P'}[Z]}$. We have that
  $E_Q[f]\leq 0$ for all $f\in\overline{C_p(P')}$. In particular,
  $\pm\ind_B(S^{\alpha_i}_u-S^{\alpha_i}_t)\in C_p(P')$ for
  $t\leq u\in[0,1]$ and
  $A=\{\alpha_1,\dots,\alpha_l\}\in \bigcup_{n\geq 1}\mathcal{A}^n$,
  $1\leq i\leq l$ and $B\in\mathcal{F}^A_{t}$. Hence we get
  $E_Q[\ind_B(S^{\alpha_i}_u-S^{\alpha_i}_t)]=0$, $i=1,\dots,l$ and so
  \begin{equation}\label{part}E_Q[\mathbf{S}^A_{u}|
    \mathcal{F}^A_{t}]=
    E_Q[\mathbf{S}^A_{t}| \mathcal{F}^A_{t}]\end{equation}
  almost surely. This shows the first direction of the theorem.

  Concerning the other direction, let $Q \in \mathcal{M}^q $. By the
  definition of $\mathcal{M}^q$ there thus exists some
  $P' \in \mathcal{P}_p$ such that
  $\frac{dQ}{dP'} \in L^q(\mathcal{G}, P')$.  Assume now that
  \eqref{NA} does not hold for this $P'$ and the dual $p$. Then there
  exists $f\neq 0$,
  $f\in \overline{C_p(P')}\cap L^p_+(\mathcal{G},P')$. By definition
  $f=\lim_{n\to\infty}f^n$ where the limit is in
  $L^p(\mathcal{G}, P')$ and $f^n=X^n_1-h^n$ with
  $h^n\in L^p_+(\mathcal{G},P')$ and $X^n_1\in K_0$. Clearly,
  $E_Q[X^n_1]=0$, hence $E_Q[f^n]\leq 0$ for all $n$. The convergence
  of $f^n$ to $f$ in $L^p(\mathcal{G}, P')$ implies that $f^n$
  converges to $f$ in $L^1(\mathcal{G},Q)$ and hence $E_Q[f]\leq
  0$. This is a contradiction to $f\geq0$ and $f\neq0$.
\end{proof}

The above fundamental theorem can also be reformulated in the
following way, showing that in a one filtration setting the current no
arbitrage condition is equivalent to the NFLVR condition of
\cite{DS:94} in the case of small financial markets and to the NAFLVR
condition of \cite{CuchieroKleinTeichmann} for large financial
markets, whenever there exists an equivalent martingale measure for
all $S^i$ (this is the case, for example, if all $S^i$ are bounded).
Note in particular that the result therefore only depends on the
equivalence class of $P$ but not on $P$ itself.

\begin{corollary}\label{ftap-1}
  The condition (NAFLp) holds for some $1\leq p<\infty$ if and only if
  $\mathcal{M}^1\neq\emptyset$.
\end{corollary}

\begin{remark}\label{qexists}
  If $\mathcal{M}^1\neq\emptyset$ then there always exists some $q>1$
  such that $\mathcal{M}^q\neq\emptyset$. Indeed, take any
  $Q\in\mathcal{M}^1$ and let $P'=Q$.  Then $Q\in\mathcal{M}^{\infty}$
  as $\frac{dQ}{dP'}=1$ and all $S^i \in L^1(\mathcal{G},
  Q)$. 
\end{remark}

\begin{proof}
  If (NAFLp) holds, then by Theorem~\ref{ftap-p}
  $\mathcal{M}^q\neq\emptyset$. And, clearly
  $\mathcal{M}^q\subseteq\mathcal{M}^1$. The reverse direction follows
  by Remark~\ref{qexists} as (NAFL1) holds for $Q\in\mathcal{M}^1$.
\end{proof}

\subsection{Equivalent formulations for (NAFLp)}\label{sec:equivform}

In the spirit of Remark \ref{stronger_topologies} one can introduce a
slightly weaker versions of (NAFLp) by considering the projective
locally convex topology on
$ \cap_{P' \in \mathcal{P}_p} L^p(\mathcal{G}, P') $. For a set
$E \in \cap_{P' \in \mathcal{P}_p} L^p(\mathcal{G}, P') $ we denote
the closure with respect to this topology by $\overline{E}^{\cap}$.

\begin{corollary}\label{cor:equiv}
  The following conditions are equivalent.
  \begin{equation}\label{NA_cap}
    \begin{split}
      \text{(NAFLp)} &\Leftrightarrow   \bigcap_{P' \in \mathcal{P}_p} \overline{C_p(P')} \cap \bigcap_{P' \in \mathcal{P}_p} L_+^p(\mathcal{G}, P') =\{0\} \\
      &\Leftrightarrow \overline{\bigcap_{P' \in \mathcal{P}_p}
        C_p(P')}^{\cap}\cap \bigcap_{P' \in \mathcal{P}_p}
      L_+^p(\mathcal{G}, P')=\{0\} \, .
    \end{split}
  \end{equation}
\end{corollary}

\begin{remark}
  Note that
  $\overline{\bigcap_{P' \in \mathcal{P}_p}
    C_p(P')}^{\cap}=\overline{\mathcal{C}_p}^{\cap}$ where
  $\mathcal{C}_p$ was introduced in Remark \ref{cp=k0minuslp} (ii).
\end{remark}

\begin{proof}
  Indeed, the implications
  \begin{align*}
    \text{(NAFLp)} &\Rightarrow \bigcap_{P' \in \mathcal{P}_p} \overline{C_p(P')} \cap \bigcap_{P' \in \mathcal{P}_p} L_+^p(\mathcal{G}, P') =\{0\}\\& \Rightarrow \overline{\bigcap_{P' \in \mathcal{P}_p} C_p(P')}^{\cap}\cap \bigcap_{P' \in \mathcal{P}_p} L_+^p(\mathcal{G}, P')=\{0\}
  \end{align*}
  hold since
  \[
    \overline{C_p(P')} \supseteq \bigcap_{P' \in \mathcal{P}_p}
    \overline{C_p(P')} \supseteq \overline{\bigcap_{P' \in
        \mathcal{P}_p} C_p(P')}^{\cap}.
  \]
  and since we can replace $L^p_+(\mathcal{G},P)$ by
  $\cap_{P' \in \mathcal{P}_p} L_+^p(\mathcal{G}, P')$ in the
  Definition of (NAFLp) (see Remark \ref{cp=k0minuslp} (ii)).  In
  order to prove that the last condition implies (NAFLp) we apply the
  Hahn-Banach theorem, in this locally convex case to construct an
  element $ Q \in \mathcal{M}^q $, i.e.~a normalized, separating
  continuous linear functional, which -- by an exhaustion argument --
  maps characteristic functions $ 1_A $ for measurable sets $ A $ with
  $ P(A) > 0 $ to positive numbers (compare with the proof of Theorem
  \ref{th:asy} in \cite{Y:80,S:90}). Note here that the relevant fact
  used here is that $\cup_{P' \in \mathcal{P}_p} L^q(\mathcal{G}, P')$
  is the dual of $\cap_{P' \in \mathcal{P}_p} L^p(\mathcal{G}, P')$ as
  shown in Remark \ref{stronger_topologies}. The existence of a
  measure $Q \in \mathcal{M}^q$ means by Theorem~\ref{ftap-p} that
  (NAFLp) holds true.
\end{proof}

The bipolar theorem now allows to show equality of the following sets
\[
  \overline{\bigcap_{P' \in \mathcal{P}_p} C_p(P')}^{\cap}=\bigcap_{P'
    \in \mathcal{P}_p} \overline{C_p(P')}
\]
under (NAFLp), yielding a nice characterization of the closure of
$\cap_{P' \in \mathcal{P}_p} C_p(P')$ in the projective locally convex
topology.  To this end, let us introduce the polar cone of a convex
cone $E \in \cap_{P' \in \mathcal{P}_p} L^p(\mathcal{G},P')$ denoted
by $E^{\circ}$:
$$E^{\circ}=\left\{g\in \bigcup_{P'\in \mathcal{P}_p} L^{q}(\mathcal{G},P'): E[fg]\leq 0,\text{ for all $f\in E$}\right\}.$$

\begin{theorem}\label{th:sameset}
  Under (NAFLp) (or one of the equivalent conditions in
  \eqref{NA_cap}) it holds that
  \[
    \overline{\bigcap_{P' \in \mathcal{P}_p}
      C_p(P')}^{\cap}=\bigcap_{P' \in \mathcal{P}_p}
    \overline{C_p(P')}.
  \]
\end{theorem}

\begin{proof}
  Let us show first that
  \[
    \big(\underbrace{\overline{\bigcap_{P' \in \mathcal{P}_p}
        C_p(P')}^{\cap}}_{:=V}\big)^{\circ}=\big(\underbrace{\bigcap_{P'
        \in \mathcal{P}_p} \overline{C_p(P')}}_{:=W}\big)^{\circ} =
    \bigcup_{\lambda\geq 0}\la\overline{\mathcal{M}^{q}}^{\cup},
  \]
  where $\overline{\mathcal{M}^{q}}^{\cup}$ denotes closure with
  respect to the injective locally convex topology on
  $ \cup_{P' \in \mathcal{P}_p} L^q(\mathcal{G}, P') $.  First assume
  that $Z=\frac{dQ}{dP'}$ for some
  $Q\in\overline{\mathcal{M}^{q}}^{\cup}$ and $P' \in
  \mathcal{P}_p$. Let $f\in V, W$.  Then $f \in \overline{C_p(P')}$
  and thus $E_Q[f]\leq0$.  This shows
  $\bigcup_{\lambda\geq
    0}\la\overline{\mathcal{M}^{q}}^{\cup}\subseteq {V}^{\circ},
  {W}^{\circ}$.

  Assume now $Z\in {V}^{\circ}, {W}^{\circ}$. As
  $-\cap_{P' \in \mathcal{P}_p} L^p_+(\mathcal{G}, P')\subseteq V,W$
  this immediately implies that $Z\geq 0$ a.s. Assume the non-trivial
  case that $P(Z>0)>0$ and define a probability measure $Q\ll P'$ for
  some $P' \in \mathcal{P}_p$ via
  $\frac{dQ}{dP'}=\frac{Z}{E_{P'}[Z]}$. Hence we get $E_Q[f]\leq0$ for
  all $f\in V,W$. 
  As all $S^i_t$ are in
  $\cap_{P' \in \mathcal{P}_p} L^p(\mathcal{G}, P')$ we have that, for
  $t\leq u$, $\pm\ind_{B}(S^i_u-S^i_{t})\in V,W$ for $B$ in an
  appropriate $\mathcal{F}^i_t$ and all $i\in I$. This shows that
  $Q\in\overline{\mathcal{M}^{q}}$ and proves the above claim. By the
  bipolar theorem applied in this locally convex case, we then have
  \[
    V^{\circ \circ}=V, \quad W^{\circ \circ}=\overline{W}^{\cap}.
  \]
  As the polars $V^{\circ}=W^{\circ}$ it follows that
  $V=\overline{W}^{\cap}$. But since
  $W \supseteq V= \overline{W}^{\cap}$ it follows that
  $W=\overline{W}^{\cap}=V$.
\end{proof}

\begin{remark}
  Working with the topologically more involved setting of
  intersections and unions of $L^p$- and $L^q$-spaces, we see that
  (NAFLp) could have been defined via
  \[
    \bigcap_{P' \in \mathcal{P}_p} \overline{C_p(P')} \cap \bigcap_{P'
      \in \mathcal{P}_p} L_+^p(\mathcal{G}, P') =\{0\}.
  \]
  In case of bounded price processes this is already on the level of
  the ``no-arbitrage'' condition similar to (NFLVR) (in a one
  filtration setup with finitely many assets), however without an
  (explicit) admissibility assumption which is however implicit due to
  our simple bounded trading strategies.
\end{remark}

\subsection{No asymptotic $L^p$ free lunch for long only portfolios
  and FTAP}\label{main_weakened}

We will now make the setting of Remark~\ref{weakerass} precise and
proceed in an analogous way. Indeed, in this subsection we will assume
the following.

\begin{assumption}\label{lplong}
  We assume that there is an equivalent probability measure
  $P' \sim P$ such that
  \begin{equation*}
    (S^i_u-S^i_t)^-\in L^p(\mathcal{G},P'),\text{ for all $i\in I$,
      $t\leq u\in[0,1]$ }
  \end{equation*}
  for some fixed $1 \leq p < \infty$.  We denote the set of all
  measures $P' \sim P$ satisfying this property by
  $\mathcal{P}_p^{\text{long}} $.
\end{assumption}

Let us define the set of all equivalent $L^q$-measures such that the
optional projections of each $S^i$ are supermartingales:

\begin{align*}
  \mathcal{S}^q =\{&Q\sim P \, | \, \exists P' \in \mathcal{P}_p^{\text{long}}  \text{ s.t. } \frac{dQ}{dP'}\in L^q(\mathcal{G},P') \text{ and } E_Q[S^{\alpha_i}_u|\mathcal{F}^A_t]\leq
                     E_Q[S_t^{\alpha_i}|\mathcal{F}^A_t]\text{ a.s.,}\\ &\text{for all }A=\{\alpha_1,\dots,\alpha_l\}\in\bigcup_{n\geq 1}\mathcal{A}^n, 1\leq i\leq l\text{ and all }t\leq u\in[0,1]\}
\end{align*}

In the definition of the wealth processes $\mathcal{X}^A$ as in
Definition~\ref{XA}, for all sets $A$, we assume now that,
additionally, $H^{\alpha_j}_{t_{i-1}}\geq0$, for all $i,j$. This means
that we are only allowed to have long positions in all assets. The
corresponding definitions of $K_0$, $C_p(P')$, $\overline{C_p(P')}$
and (NAFLp) are then analogous as in Definition~\ref{largeX1} and
Definition~\ref{CpNAFLp}.

Note that by the Assumption \ref{lplong} it is clear that
$C_p(P')\neq\emptyset$. Indeed, for $f\in K_0$ defined with a {\it
  bounded} nonnegative $\mathbb{F}$-simple integrand we have that, for
example,
$$f\wedge 1=f-(f-1)\ind_{\{f>1\}}\in C_p(P'),$$
as, by the boundedness of the integrand and Assumption \ref{lplong} we
have that $(f\wedge1)^{-}\in L^p(\mathcal{G}, P')$ and by definition
$f\wedge 1\leq1$.

\begin{theorem}\label{variant}
  Suppose that Assumption \ref{lplong} holds for some fixed
  $1 \leq p< \infty$. Then (NAFLp) holds with non-negative strategies
  if and only if $\mathcal{S}^q \neq \emptyset$ where $q$ satisfies
  $\frac {1}{p} + \frac{1}{q}=1$.
\end{theorem}

\begin{proof}
  Assume that (NAFLp) holds. We proceed exactly as in the proof of
  Theorem~\ref{ftap-p}. But in the last step of this direction we only
  get that
  $\ind_B(S^{\alpha_i}_u-S^{\alpha_i}_t)^-\in L^p(\mathcal{G},P')$
  which immediately implies that
  \begin{align*}
    &\left(\ind_B(S^{\alpha_i}_u-S^{\alpha_i}_t)\right)\wedge n\\
    &\quad=\ind_B(S^{\alpha_i}_u-S^{\alpha_i}_t)-\left(\ind_B(S^{\alpha_i}_u-S^{\alpha_i}_t)-n\right)\ind_{\{\ind_B(S^{\alpha_i}_u-S^{\alpha_i}_t)>n\}}
      \in C_p(P'),
  \end{align*}
  for all $n\geq1$. Hence
  $E_Q[(\ind_B(S^{\alpha_i}_u-S^{\alpha_i}_t))\wedge n]\leq0$, for all
  $n\geq 1$. By Fatou's Lemma we get that
  $E_Q[\ind_B(S^{\alpha_i}_u-S^{\alpha_i}_t)]\leq0$ which shows the
  first direction of the theorem.

  Concerning the other direction, let $Q \in \mathcal{S}^q $. By the
  definition of $\mathcal{S}^q$ there thus exists some
  $P' \in \mathcal{P}^{\text{long}}_p$ such that
  $\frac{dQ}{dP'} \in L^q(\mathcal{G}, P')$.  Assume now that
  \eqref{NA} (for non-negative strategies) does not hold for this $P'$
  and the dual $p$.  Proceeding exactly as in the proof of
  Theorem~\ref{ftap-p}, due to the non-negativity of the strategies, the boundedness of the integrands and as $Q\in \mathcal{S}^q$
  we
   get $E_Q[X^k_1]\leq0$. The rest
  follows analogously as in the proof of Theorem~\ref{ftap-p}.
\end{proof}

\section{A Super-replication result}\label{sec:super-replication}
This section is dedicated to present a super-replication results in the present 
$L^p$-setting.  Throughout this section we assume that, for some fixed
$1 \leq p < \infty$, \eqref{NA} holds for the original measure $P$,
and we say (NAFLp) holds for $P$ and write $C_p$ for $C_p(P)$. This
means in particular that $P \in \mathcal{P}_p$. Let us also introduce
the following sets of measures
\begin{equation*}
  \begin{split}
    \mathbb{M}^q=\mathcal{M}^q(P) =\{&Q\sim P \, |\, \frac{dQ}{dP}\in
    L^q(\mathcal{G},P) \text{ and
    }E_Q[S^{\alpha_i}_u|\mathcal{F}^A_t]=
    E_Q[S^{\alpha_i}_t|\mathcal{F}^A_t]\text{ a.s.,}\\ &\text{for all
    }A=\{\alpha_1, \ldots, \alpha_l\} \in\bigcup_{n\geq
      1}\mathcal{A}^n, 1\leq i\leq l\text{ and all }t\leq u\in[0,1]\},
  \end{split}
\end{equation*}
which play a key role in our super-replication result.  Note that
$\mathcal{M}^q= \bigcup_{P' \in \mathcal{P}_p} \mathcal{M}^q(P')$ and
that the proof of Theorem \ref{ftap-p} implies the following
assertion.

\begin{corollary}
  The condition (NAFLp) holds for $P$ if and only if
  $\mathbb{M}^q \neq \emptyset$.
\end{corollary}

We henceforth identify measures $Q\in\mathbb{M}^q$ with their density
$\frac{dQ}{dP}$ so that we can consider $\mathbb{M}^q$ as a subset of
$L^{q}(\mathcal{G},P)$. Recall that $\overline{\mathbb{M}^q}$ is the
closure of $\mathbb{M}^q$ in $L^q(\mathcal{G},P)$.

\begin{remark}
  The closure of $\mathbb{M}^q$ in $L^q$ just consists of the
  corresponding absolutely continuous measures, that is,
  \begin{align}
    \overline{\mathbb{M}^q} &=\{Q\ll P, \frac{dQ}{dP}\in L^q(\mathcal{G},P): E_Q[S^{\alpha_i}_u|\mathcal{F}^A_t]=
                              E_Q[S_t^{\alpha_i}|\mathcal{F}^A_t]\nonumber\text{ a.s.,}\nonumber\\ &\text{for all }A=\{\alpha_1,\dots,\alpha_l\}\in\bigcup_{n\geq 1}\mathcal{A}^n, 1\leq i\leq l\text{ and all }t\leq u\in[0,1]\}
  \end{align}
  Indeed, take any $Q\in\overline{\mathbb{M}^q}$ and
  $Q'\in \mathbb{M}^q$. Then
  $Q^n=(1-2^{-\frac{n}{q}})Q+2^{-\frac{n}{q}}Q'\in \mathbb{M}^q$
  converges to $Q$ with respect to $L^q(\mathcal{G}, P)$-norm.
\end{remark}

Analogously to Section \ref{sec:equivform}, denote by $E^{\circ}$ the
polar cone of a convex cone $E\subseteq L^p(\mathcal{G},P)$ , i.e.,
$$E^{\circ}=\{g\in L^{q}(\mathcal{G},P): E[fg]\leq 0,\text{ for
  all $f\in E$}\}.$$
	
Similarly as in the proof of Theorem \ref{th:sameset}, we can show the
following duality result.
	
\begin{lemma}\label{duality}
  For the polar cone the following identity holds true
  \[ {(\overline{C_p})}^{\circ}={(C_p)}^{\circ}=\bigcup_{\lambda\geq
      0}\la\overline{\mathbb{M}^{q}} \, .
  \]
\end{lemma}

\begin{proof}[Proof of Lemma~\ref{duality}]
  First assume that $Z=\frac{dQ}{dP}$ for some
  $Q\in\overline{\mathbb{M}^{q}}$. Let $f\in \overline{C_p}$.  So
  $f=\lim_{n\to\infty}f^n$ in $L^p(\mathcal{G}, P)$ with $f^n\in
  C_p$. Hence $E_Q[f^n]\leq0$ for all $n$ and the same holds for the
  $L^1(\mathcal{G}, P)$-limit of $Zf^n$, and so $E_Q[f]\leq0$.  This
  shows
  $\bigcup_{\lambda\geq 0}\la\overline{\mathbb{M}^{q}}\subseteq
  {(\overline{C_p})}^{\circ}$. And clearly
  ${(\overline{C_p})}^{\circ}\subseteq {(C_p)}^{\circ}$.

  Assume now $Z\in {(C_p)}^{\circ}$. As
  $-L^p_+(\mathcal{G}, P)\subseteq C_p$ this immediately implies that
  $Z\geq 0$ a.s. Assume the non-trivial case that $P(Z>0)>0$ and
  define a probability measure $Q\ll P$ via
  $\frac{dQ}{dP}=\frac{Z}{E_{P}[Z]}$. Hence we get $E_Q[f]\leq0$ for
  all $f\in C_p$ and in the $L^p(\mathcal{G}, P)$-closure
  $\overline{C_p}$. As all $S^i_t$ are in $L^p(\mathcal{G}, P)$ as
  $P \in \mathcal{P}_p$ we have that, for $t\leq u$,
  $\pm\ind_{B}(S^i_u-S^i_{t})\in C_p$ for $B$ in an appropriate
  $\mathcal{F}^i_t$ and all $i\in I$. This shows that
  $Q\in\overline{\mathbb{M}^{q}}$. This finishes the proof.
\end{proof}

We can now prove the following super-replication result.

\begin{theorem}\label{super}
  Let $f\in L^p(\mathcal{G},P)$. Then
$$\sup_{Q\in\mathbb{M}^{q}}E_Q[f]=\inf \big \{x \in \mathbb{R} \, | \; \exists \,
g\in \overline{C_{p}}\text{ with } x+g \geq f \big \}$$
\end{theorem}

\begin{proof}[Proof of Theorem~\ref{super}]
  By Lemma~\ref{duality} it is clear that $\sup \leq \inf$. Let now
  $x_0=\sup_{Q\in\mathbb{M}^q}E_Q[f]$ and suppose that
  $x_0<\inf \{x \in \mathbb{R} \, | \; \exists \, g\in
  \overline{C_{p}}\text{ with } x+g \geq f \big \}$.  Then
  $f-x_0\notin \overline{C_p}$. Hence there exists
  $Z\in L^{q}(\mathcal{G},P)$ such that
  $\sup_{g\in \overline{C_p}}E[Zg]\leq 0$ and $E[Z(f-x_0)]>0$. This
  implies that $Z\in (\overline{C_p})^{\circ}$ and $Z\ne0$ therefore
  we can define a measure $Q\in\mathbb{M}^{q}$ by
  $\frac{dQ}{dP}=\frac{Z}{E_{P}[Z]}$. We get
$$x_0<E_Q[f]\leq\sup_{R\in\mathbb{M}^{q}}E_R[f].$$ This
implies that $\inf\leq\sup$.
\end{proof}

\begin{remark}
  As shown in the course of the proof in Theorem \ref{th:sameset}
  \[ {\Big(\overline{\bigcap_{P' \sim \mathcal{P}_p}
        C_p(P')}^{\cap}\Big)}^\circ ={\Big(\bigcap_{P' \sim
        \mathcal{P}_p} \overline{C_p(P')}\Big)}^\circ=
    \bigcup_{\lambda\geq 0}\la\overline{\mathcal{M}^{q}}^{\cup}
  \]
  holds true. Recall here that $\overline{E}^{\cap}$ and
  $\overline{E}^{\cup}$ denote the closure of a set $E$ with respect
  to the projective respectively injective locally convex topology on
  $ \cap_{P' \in \mathcal{P}_p} L^p(\mathcal{G}, P') $ respectively
  $ \cup_{P' \in \mathcal{P}_p} L^q(\mathcal{G}, P') $. This gives
  rise to another slightly different super-replication result, namely
  for $f\in \cap_{P' \in \mathcal{P}_p} L^p(\mathcal{G},P')$, we have
$$\sup_{Q\in\mathcal{M}^{q}}E_Q[f]=\inf \left\{x \in \mathbb{R} \, | \; \exists \,
  g\in \bigcap_{P' \in \mathcal{P}_p}\overline{C_p(P')}\text{ with }
  x+g \geq f \right\}$$
	
\end{remark}

In the setting of Subsection~\ref{main_weakened} under the weaker
Assumption~\ref{lplong} we get an analogous super-replication result
for long only strategies and measures in $\mathbb{S}^Q$ defined as
follows:

\begin{align*}
  \mathbb{S}^q =\{&Q\sim P \, | \,  \frac{dQ}{dP}\in L^q(\mathcal{G},P) \text{ and } E_Q[S^{\alpha_i}_u|\mathcal{F}^A_t]\leq
                    E_Q[S_t^{\alpha_i}|\mathcal{F}^A_t]\text{ a.s.,}\\ &\text{for all }A=\{\alpha_1,\dots,\alpha_l\}\in\bigcup_{n\geq 1}\mathcal{A}^n, 1\leq i\leq l\text{ and all }t\leq u\in[0,1]\}.
\end{align*}

\begin{theorem}\label{variantsuper}
  Let $f\in L^p(\mathcal{G},P)$. Then
$$\sup_{Q\in\mathbb{S}^{q}}E_Q[f]=\inf \big \{ x \in \mathbb{R} \, | \; \exists \,
g\in \overline{C_{p}}\text{ for long only strategies such that } x+g
\geq f \big \} \, .$$
\end{theorem}

\begin{proof} For the proof of Theorem~\ref{variantsuper} we have to
  adapt Lemma~\ref{duality} by replacing $\overline{\mathbb{M}^q}$ by
  $\overline{\mathbb{S}^q}$. In the proof we get by the boundedness of the integrands and as $Q\in \overline{\mathbb{S}^q}$ that $E_Q[f^n]\leq0$
  for $f^n=X^n_1-h^n$ with $f^n\in L^p(\mathcal{G},P)$, $X^n_1\in K_0$ and
  $h^n\in L^0_+(\mathcal{G},P)$.  The rest is identical.
\end{proof}

The next theorem represents elements of $ \overline{C_p} $ as
$L^0$-limits of replicable claims minus consumptions, and clarifies
additionally that
\[
  \overline{C_p} \cap - \overline{C_p} \subset \overline{K_0}^{L^0}
  \cap L^p(\mathcal{G}, P)
\]
as well as that proper intervals of arbitrage-free prices are
open. Here we denote by $ \overline{K_0}^{L^0}$ the $L^0$-closure of
$K_0$. These considerations are of course almost classical, their
proofs do not differ much from classical counterparts.

\begin{remark}\label{attainable}
  \begin{enumerate}
  \item Note that
    $\overline{C_p} \cap - \overline{C_p}\supseteq \overline{K_0}$,
    where $\overline{K_0}$ is the $L^p$-closure of $K_0$.
  \item The set $\overline{C_p} \cap - \overline{C_p}$ is dually
    characterized as the set of elements $ g $ such that
    $ E_Q[g] = 0 $ for all $ Q \in \overline{\mathbb{M}^q }$ by the
    bipolar theorem.
  \end{enumerate}
  For (ii) note that if $g$ is in
  $\overline{C_p} \cap - \overline{C_p}$. Then clearly $E_Q[g]=0$ for
  all $Q\in\overline{\mathbb{M}^q }$. On the other hand, suppose
  $E_Q[g]=0$ for all $Q\in\overline{\mathbb{M}^q }$. Then $g$ and $-g$
  are in
  $\left(\bigcup_{\lambda\geq0}\lambda\overline{\mathbb{M}^q
    }\right)^{\circ}=\overline{C_p}$, where the last equality holds by
  the bipolar theorem.
\end{remark}

\begin{theorem}\label{actual_super-replication}
  Assume that (NAFLp) holds for $P$.
  \begin{enumerate}
  \item Every $g\in\overline{C_p} $ can be represented as $L^p$-limit
    $g=\lim_{n\to \infty} \big(f^n - h^n) = f - h $, where $f$ and $h$
    are finitely valued random variables, with
    $ f = \lim_{n \to \infty} f^n $ appearing only as limit in
    probability of a sequence $ f_n \in K_0 $, and
    $ \lim_{n\to \infty} h^n = h \geq 0 $ being again a limit in
    probability of finitely valued, non-negative random variables
    $h^n$.
  \item Let $\tilde{g}\in L^p(\mathcal{G},P)$. Then either
    $\tilde{g} $ is replicable (attainable), i.e.
  $$ \tilde{g} -x \in \overline{C_p} \cap - \overline{C_p}= \bigcap_{Q\in\overline{\mathbb{M}^q}}\overline{K_0}^{L^1(Q)}\cap L^p(\mathcal{G},P)
  \subset \overline{K_0}^{L^0} \cap L^p(\mathcal{G},P)$$ for some
  $ x \in \mathbb{R}$, or there are at least two measures
  $ Q, Q' \in \mathbb{M}^q $ such that
  $ E_Q[\tilde{g}] \neq E_{Q'}[\tilde{g}] $. In the second case the
  super-replication price
  $ x = \sup_{Q\in\mathbb{M}^{q}}E_Q[\tilde{g}]$ is not attained by
  any equivalent measure $ Q \in \mathbb{M}^q $.
\end{enumerate}
\end{theorem}

\begin{proof}
  For the first assertion take
  $ g = \lim_{n \to \infty}( g^n - k^n )$, an $L^p$-limit, where
  $ g^n \in K_0 $ and $ k^n \in L^p_+(\mathcal{G},P)$. By Komlos'
  theorem we can choose forward convex combinations $ h^n $ of
  elements $ k^n,k^{n+1},\ldots $ such that $ h^n \to h $ in
  probability, where $ h\geq0 $ is a not necessarily finitely valued
  random variable. Take forward convex combinations with the same
  weights of $ g^n,g^{n+1},\ldots $ and denote them by $ f^n $. Then
  again $ f^n - h^n \to g $ in $ L^p $. Take now any
  $ Q \in \mathbb{M}^q $. Then $E_Q[f^n]=0$ and by
  $L^p(\mathcal{G},P)$-convergence we have that
  $\lim_{n\to\infty}E_Q[-h^n]=E_Q[g]>-\infty$.  By Fatou's Lemma
  $$ 0 \leq E_Q[h] \leq \lim_{n \to \infty} E_Q[h^n]<\infty ,$$ whence $ h $ is
  finitely valued as well as the limit
  $ \lim_{n \to \infty} f^n = f $, which is only understood in $L^0$.

  For the second assertion take $ \tilde g \in L^p(\mathcal{G},P)$ and
  $ x \in \mathbb{R} $ such that
  $ g=\tilde{g} - x \in \overline{C_p} \cap - \overline{C_p} $. By
  Remark~\ref{attainable}, $E_Q[g]=0$ for all
  $Q\in\overline{\mathbb{M}^q}$.  As in the previous step by passing
  to forward convex combinations we find two sequences $ f^n \in K_0 $
  and $h^n \geq 0 $, each converging in $L^0$ to finitely valued
  random variables, such that $ f^n -h^n \to g $ in $ L^p $ as
  $ n \to \infty $. Take any $ Q \in \mathbb{M}^q $, then
  $0\leq\lim_{n\to\infty}E_Q[h^n]= -E_Q[g]=0 $ by convergence in
  $L^p$. Hence actually $ h = 0 $. Therefore
  $ g \in \overline{K_0}^{L^0}\cap L^p(\mathcal{G},P)$. Moreover we
  have that $h^n\to0$ in $L^1(\mathcal{G},Q)$ and hence
   $$E_Q[|f^n-g|]\leq E_Q[|f^n-h^n-g|]+E_Q[h^n]\to0,$$
   for $n\to\infty$ as $f^n-h^n\to g$ in $L^p(\mathcal{G},P)$ and
   $h^n\to0$ in $L^1(\mathcal{G},Q)$. So we have that, in fact,
   $g\in
   \bigcap_{Q\in\overline{\mathbb{M}^q}}\overline{K_0}^{L^1(Q)}$.

   This argument holds also true when there is only one element
   $ Q \in \mathbb{M}^q $ such that $ E_Q[\tilde g] = x $, whence the
   last assertion that proper pricing intervals have to be open which
   in turn only appears in case of non-attainability.

   Finally observe that for
   $g\in
   \bigcap_{Q\in\overline{\mathbb{M}^q}}\overline{K_0}^{L^1(Q)}\cap
   L^p$ we have that $E_Q[g]=0$ for any $Q\in\overline{\mathbb{M}^q}$
   and hence $g$ is in $\overline{C_p} \cap - \overline{C_p}$ which
   shows the equality of the two sets.
 \end{proof}

\begin{remark}
  Notice that a replicable claim in our setting is replicated by
  $L^0$-limits of elements of $K_0$, not necessarily by its
  $L^p$-limits. This subtlety cannot be removed.
\end{remark}

The following example shows that
$ \overline{C_1} \cap - \overline{C_1} \subsetneq \overline{K_0}^{L^0}
\cap L^1(\mathcal{G}, P)$:

\begin{example}
  Consider a one period market with countably many derivatives
  $ f^j \geq -1 $ at time $ T=1$, $ j \geq 0 $, at price $0$ at time
  $ t=0 $. Assume for simplicity that the historical measure $P$
  satisfies already the following two conditions:
  \begin{itemize}
  \item $ E[f^j]=0 $ for $ j \geq 0 $.
  \item The sequence $ f^j $ converges to $ - 1 $ almost surely with
    respect to $P$, hence -- of course -- the convergence is not in
    $ L^1 $.
  \end{itemize}

  In this case we can create, as in the introduction, a two filtration
  setting, where hedging is actually buy \& hold in this large
  financial market, and the filtration $ \mathbb{F} $ is trivial. In
  this setting $ K_0 $ contains all (finite) linear combinations of
  $ f^j $, its closure in $L^1$ only contains elements with vanishing
  expectations, however, its $L^0$ closure even if intersected with
  $L^1$ contains the constant function $ -1 $ and $1$ (the latter by
  taking $-f_j\in K_0$).  However
  $ \{-1,1\} \notin \overline{C_1} \cap - \overline{C_1}$, whence
  \[
    \overline{C_1} \cap - \overline{C_1} \subsetneq
    \overline{K_0}^{L^0} \cap L^1.
  \]
\end{example}

\appendix
\section{A technical result}

The following theorem, which goes back to Jia-An Yan \cite{Y:80} for $p=1$
and to Jean-Pascal Ansel for the case $1\leq p < \infty$ is taken from
\cite{S:90}:
\begin{theorem}\label{th:asy}
  Let $E$ be a convex subset of $L^p(\mathcal{G},P)$ with $0 \in E$. 
  Then the following three conditions are equivalent:
  \begin{enumerate}
  \item For every $\eta \in L_+^p(\mathcal{G},P)$, $\eta \neq 0$,
    there exits some $c >0$ such that
    $c \eta \notin \overline{E-L^p_+(\mathcal{G},P)}$.
  \item For every $A \in \mathcal{G}$ such that $P[A] >0$, there exists some
    $c >0$ such that
    $c1_{A} \notin \overline{E-L^p_+(\mathcal{G},P)}$.
  \item There exists a random variable $Z \in L^q(\mathcal{G}, P)$
    such that $Z >0$ a.s. and $\sup_{Y \in E} E[ZY]< \infty$.
  \end{enumerate}
\end{theorem}

\end{document}